\documentclass[10pt]{article}
\usepackage[square,sort,comma,numbers]{natbib}
\usepackage[T1]{fontenc}
\usepackage[letterpaper,margin=1in]{geometry}
\usepackage{amsmath,amssymb,amsthm,mathtools}
\usepackage{booktabs}
\usepackage{array}
\usepackage{enumitem}
\usepackage{microtype}
\usepackage{xcolor}
\usepackage{framed}
\definecolor{shadecolor}{gray}{0.93}

\usepackage{tikz}
\usetikzlibrary{positioning}
\usepackage{float}
\usepackage{hyperref}
\usepackage[ruled,linesnumbered,vlined]{algorithm2e}

\hypersetup{
  colorlinks=true,
  linkcolor=red,
  citecolor=green!55!black,
  urlcolor=magenta,
  pdftitle={Non-Existence of Exact Pairwise Maximin Share Allocations for Additive Goods and Chores},
  pdfauthor={Xiaohui Bei, Zehan Lin, Shengxin Liu, Rong Luan, Biaoshuai Tao}
}

\newtheorem{theorem}{Theorem}
\newtheorem{lemma}[theorem]{Lemma}
\newtheorem{proposition}[theorem]{Proposition}

\newtheorem{remark}[theorem]{Remark}

\newcommand{\PMMS}{\mathrm{PMMS}}
\newcommand{\PiTwo}{\Pi_2}
\newcommand{\GoodsMu}{\mu^{\mathrm{G}}}
\newcommand{\ChoreMu}{\mu^{\mathrm{C}}}
\newcommand{\alphaPMMS}
{\alpha^{\star}_{\PMMS}}
\newcommand{\rhoPMMS}{\rho^{\star}_{\PMMS}}

\title{Non-Existence of PMMS Allocations and\\
a $4/3$-PMMS Guarantee for Additive Chores}
\author{%
  \begin{tabular}{@{}c@{}}
    Xiaohui Bei\textsuperscript{1}, Zehan Lin\textsuperscript{2}, Shengxin Liu\textsuperscript{3},
    Rong Luan\textsuperscript{1}, and Biaoshuai Tao\textsuperscript{4} \\[1.0em]
    {\normalsize\textsuperscript{1}Nanyang Technological University, \texttt{\{xhbei@ntu.edu.sg, luan0020@e.ntu.edu.sg\}}} \\
    {\normalsize\textsuperscript{2}University of Macau, \texttt{yc47490@um.edu.mo}} \\
    {\normalsize\textsuperscript{3}Harbin Institute of Technology, Shenzhen, \texttt{sxliu@hit.edu.cn}} \\
    {\normalsize\textsuperscript{4}Shanghai Jiao Tong University, \texttt{bstao@sjtu.edu.cn}}
  \end{tabular}%
}
\date{}

\begin{document}
\maketitle

\begin{abstract}
We study pairwise maximin share (PMMS) fairness for indivisible items with additive preferences. We give a polynomial-time reduction from chores to goods that preserves the existence of a PMMS allocation. Together with known nonexistence results for chores, this yields nonexistence for additive goods. In addition, we show that deciding if a given instance admits a PMMS allocation is NP-hard. We also give explicit instances whose PMMS factors are $226/227$ for goods and $1.102065$ for chores, certified by exact enumeration. Complementing these impossibility results, we prove that every additive-chore instance admits a $4/3$-PMMS allocation.
\end{abstract}

\section{Introduction}
% Citation keys and bibliographic metadata follow DBLP records in reference.bib.

Fair division studies how to allocate a set $M$ of $m$ items among a set $N$ of $n$ agents with heterogeneous preferences. Depending on whether the items provide value or impose costs on the agents, they are modeled as \emph{goods} or \emph{chores}, respectively.
Examples include distributing resources and assigning unwanted tasks. When items are indivisible, even basic fairness requirements can be impossible to satisfy. 
Understanding which fairness guarantees remain achievable is therefore a central question in algorithmic game theory. 
In this paper, we focus on additive preferences, under which the value or cost of a bundle is the sum of the values or costs of its items.

A fundamental fairness notion is \emph{envy-freeness} (EF), which requires that no agent prefer another agent's bundle to her own. However, an envy-free allocation need not exist even when a single item is allocated between two agents with identical preferences. This obstacle has motivated a range of relaxations, including the \emph{maximin share} (MMS) guarantee, which extends the logic of cut-and-choose to multiple agents~\cite{journals/ai/AmanatidisABFLMVW23}. An agent partitions all items into as many bundles as there are agents and evaluates the least preferred bundle. For goods, her MMS is the largest minimum bundle value she can secure over all such partitions. For chores, the analogous benchmark is the smallest maximum bundle cost. An MMS allocation gives every agent a bundle at least as desirable as her benchmark.
Introduced by Caragiannis et al.~\cite{journals/teco/CaragiannisKMPS19}, \emph{pairwise maximin share} (PMMS) applies this partition-based guarantee to pairwise comparisons. Let $X_i$ denote the bundle allocated to agent $i$. For every pair of distinct agents $i$ and $j$, agent $i$ considers all bipartitions of $X_i\cup X_j$. For goods, PMMS requires that the value of $X_i$ be at least the maximum value she can guarantee for the less valuable part. For chores, it requires that the cost of $X_i$ be at most the minimum cost she can guarantee for the more costly part. Both comparisons use agent $i$'s own preferences. Thus, PMMS combines the pairwise perspective of envy-freeness with the partition-based benchmark of MMS.

Despite this natural interpretation, the existence of PMMS allocations has proved difficult to settle. For additive goods, it remained an open question in the literature~\cite{journals/ai/AmanatidisABFLMVW23,conf/aaai/ByrkaMP26}. Byrka et al.~\cite{conf/aaai/ByrkaMP26} established existence for several restricted valuation classes and constructed a three-agent counterexample with two monotone valuations and one additive valuation.
Their counterexample does not resolve the question when all agents have additive valuations.
This leaves the following question.

\begin{center}\begin{minipage}{0.97\linewidth}\begin{shaded}\noindent\textbf{Question 1. }
%\textit{Can fair allocations of indivisible chores be computed in the comparison model using sublinear (or even logarithmic) query complexity with respect to $m$?}
%\textit{Can fair allocations of indivisible chores be computed using only logarithmically many comparison queries in the number of chores?}
\textit{Does any instance of indivisible goods with additive valuations admit a PMMS allocation?}
\end{shaded}
\end{minipage}
\end{center}

For additive chores, Sun et al.~\cite{journals/aamas/SunCD23a} showed that PMMS implies \emph{envy-freeness up to any item} (EFX), which requires that removing any positive-cost chore from an agent's own bundle eliminate her envy toward any other agent. He and Tao~\cite{journals/corr/abs-2606-08872} subsequently constructed additive chore instances with no EFX allocation.
Together, these results imply that exact PMMS allocations need not exist for additive chores. 
This implication alone, however, does not determine how closely PMMS can always be approximated. 
For chores, a factor $\rho\geq 1$ permits each agent's cost to be at most $\rho$ times her pairwise benchmark. The failure of exact PMMS therefore raises a quantitative question.

\begin{center}\begin{minipage}{0.97\linewidth}\begin{shaded}\noindent\textbf{Question 2. }
%\textit{Can fair allocations of indivisible chores be computed in the comparison model using sublinear (or even logarithmic) query complexity with respect to $m$?}
%\textit{Can fair allocations of indivisible chores be computed using only logarithmically many comparison queries in the number of chores?}
\textit{How closely can PMMS be approximated for additive chores? What lower and upper bounds can be established on the smallest approximation factor guaranteed for every instance?}
\end{shaded}
\end{minipage}
\end{center}

\subsection{Our Results}

We answer Question~1 negatively and give lower and upper bounds for Question~2. Our main results establish nonexistence for additive goods through a reduction from chores, and lower and upper approximation bounds for additive chores.

\paragraph{Nonexistence for Goods.}
Our first result follows from a polynomial-time reduction that transforms any additive-chore instance into an additive-goods instance while preserving PMMS existence (Lemma~\ref{lem:reduction_chores_goods}). Applying this reduction to a chore instance with no PMMS allocation settles the existence question for additive goods.

% \begin{resultbox}
\medskip
\noindent
\textbf{Result 1 (Theorem~\ref{thm:goods}).}
\textit{PMMS allocations need not exist for additive goods.}
% \end{resultbox}
\medskip

The reduction replaces chore costs by complementary values and adds goods to force equal bundle sizes in every PMMS allocation. We also give a separate three-agent, nine-good instance with optimal PMMS ratio $226/227$, certified by exact enumeration in Appendix~\ref{append:impossibility_ratio} (Theorem~\ref{thm:goods-ratio}).

\paragraph{Computational Hardness.}
We next consider the decision problem of whether a given additive instance admits an exact PMMS allocation. We first show that EFX and PMMS coincide under lexicographic preferences over mixed goods and chores. Combined with the NP-completeness result of \citet{hosseini2023fairly}, this implies that deciding PMMS existence is NP-complete in this domain. We then give a polynomial-time reduction from additive mixed instances to additive chore instances that preserves PMMS existence, yielding NP-hardness for additive chores. Lemma~\ref{lem:reduction_chores_goods} then subsequently transfer the hardness result to additive goods.

\medskip
\noindent
\textbf{Result 2 (Theorems~\ref{thm:NPhard-chores} and~\ref{thm:NPhard-goods}).}
\textit{It is NP-hard to decide whether an additive chores instance admits a PMMS allocation. The same decision problem is NP-hard for additive goods.}
\medskip

\paragraph{Lower and Upper Bounds for Chores.}
For chores, we complement the known failure of exact PMMS with an explicit approximation lower bound and a universal upper bound.

\medskip
\noindent
{\bf Result 3} (Theorems~\ref{thm:chores} and~\ref{thm:chores-existence}){\bf .}
{\em There exists an additive-chore instance with no $1.102$-PMMS allocation. Furthermore, every additive-chore instance admits a $4/3$-PMMS allocation.}
\medskip

The lower bound comes from a five-agent, twelve-chore instance with optimal PMMS factor $1.102065$, verified by exact enumeration in Appendix~\ref{append:impossibility_ratio}. For the upper bound, we transform a price-supported pEF1 allocation so that every agent either receives a singleton or incurs cost at most twice the cost she assigns to any other bundle. This property implies $4/3$-PMMS. Section~\ref{sec:chores-existence} gives the construction and extends it to nonnegative costs by perturbation.

\subsection{Related Work}
% We review related works that are most relevant to ours. Much of the literature on share-based fairness for indivisible items centered around the maximin-share (MMS) guarantee. The MMS notion was introduced in the context of indivisible allocation by \citet{budish2011}, and subsequently recieved substantial attention in the fair division literature.
% In particular, \citet{journals/jacm/KurokawaPW18} initiated the systematic study of MMS for additive valuations, showing that an exact MMS allocation need not always exist and establishing a $2/3$-approximation guarantee. This triggered a long line of wokr devoted to understanding the existence, computation, and approximation of MMS allocations.
Due to the vast literature on the fair allocation problem, in the following, we only review the results highly related to our work.
For a comprehensive overview of other related works, please refer to the recent surveys by \cite{journals/ai/AmanatidisABFLMVW23,journals/jair/LiuLSW24}.

\paragraph{(Approximate) PMMS Allocations for Goods.}
The notion of PMMS was first introduced later by \citet{journals/teco/CaragiannisKMPS19} in their study of the fairness and efficiency properties of maximum Nash welfare (MNW) solutions. They established that every MNW solution guarantees a $(\sqrt{5} - 1)/{2}\approx 0.618$-PMMS allocation and this golden-ratio factor is tight for the MNW rule. Subsequently, this general approximation guarantee was improved to $(\sqrt{17}-1)/4 \approx 0.7808$, providing a polynomial-time algorithm that computes an allocation achieving this factor. Until the present and concurrent works, exact PMMS existence for unrestricted additive valuations remained open. Prior work therefore focused on identifying valuation classes and allocation structures that admit exact or improved guarantees. Exact PMMS is known for binary additive and identical valuations, while agents with a common ordinal ranking admit a polynomial-time solvable $4/5$-PMMS allocation \cite{conf/aaai/BarmanBMN18,journals/tcs/DaiGMGXZ24}. For sufficiently similar valuations, \citet{conf/atal/BarmanKP24} obtained a $5/6 \gamma$-PMMS guarantee, where $\gamma$ measures the similarity between agents' item values. These results are accompanied by several implications for other fairness notions. Moreover, under additive valuations, exact PMMS implies envy-free up to any good (EFX), while EFX and envy-free up to one good (EF1) can be used to obtain approximate PMMS guarantees \cite{journals/teco/CaragiannisKMPS19,conf/ijcai/AmanatidisBM18}.
Related positive results extend to several structured valuation and allocation settings. These include matroid-rank, personalized bivalued, binary-valued, pair-demand, connected-allocation, and graphical valuation models, where exact or approximate PMMS guarantees have been established under suitable assumptions \cite{conf/atal/BarmanV21,conf/aaai/ByrkaMP26,conf/atal/Hummel024,conf/aaai/ChristodoulouM26}. In contrast, a three-agent counterexample shows that exact PMMS need not exist for general monotone valuations \cite{conf/aaai/ByrkaMP26}. These results did not settle the existence of exact PMMS for unrestricted additive valuations.

\paragraph{(Approximate) PMMS Allocations for Chores.} Compared with the goods setting, PMMS has received considerably less direct investigation for indivisible chores. The main systematic study is due to \citet{journals/aamas/SunCD23a}, who study PMMS together with MMS, EFX, and EF1 under additive and submodular cost functions. They established approximation relationships among these fairness notions and the efficiency loss induced by imposing them. In particular, they observe that every allocation is $2$-PMMS under subadditive cost functions, revealing that constant-factor approximate PMMS is substantially weaker for chores than for goods. They further show that exact PMMS implies EFX under additive costs, whereas for every $1<\alpha\leq 2$, an $\alpha$-PMMS allocation need not provide any bounded approximation to EF1 or EFX. They also establish price-of-fairness bounds for PMMS, showing that the efficiency loss can become unbounded when there are at least three agents. More recently, \citet{conf/aaai/ByrkaMP26} establish exact PMMS existence for binary-valued MMS-feasible set functions without requiring monotonicity; consequently, their result also applies to certain chore and mixed-manna instances. Beyond these results, however, the existence and algorithmic computation of PMMS allocations for general chore valuations remain comparatively underexplored.

\paragraph{Concurrent and Independent Work.}
Concurrently and independently of our work, \citet{aziz2026pmms} and \citet{golz2026pmms} also showed that PMMS allocations need not exist for additive goods. Aziz gives a counterexample with four agents and strictly positive valuations, while G{\"o}lz gives one with three agents and nine goods. G{\"o}lz further provides a computer-verified instance admitting no $\alpha$-PMMS allocation for any $\alpha>78/79$.

\section{Preliminaries}

For any positive integer $k$, let $[k]=\{1,\ldots,k\}$. Let $N=[n]$ be a finite set of agents and $M$ a finite set of indivisible items. An allocation $X=(X_i)_{i\in N}$ is a partition of $M$. We write $\PiTwo(S)$ for the set of ordered bipartitions $(P_1,P_2)$ of a bundle $S\subseteq M$.

\paragraph{Goods.}
Each agent $i$ has a nonnegative additive valuation
$v_i:2^M\to\mathbb{R}_{\ge 0}$.
For $S\subseteq M$, define agent $i$'s two-way maximin-share benchmark for goods by
\[
  \GoodsMu_i(2,S)
  :=
  \max_{(P_1,P_2)\in\PiTwo(S)}
  \min\{v_i(P_1),v_i(P_2)\}.
\]
Thus, agent $i$ partitions $S$ into two bundles and chooses a partition
that maximizes the value of the less valuable bundle.
An allocation $X=(X_1,\ldots,X_n)$ is $\alpha$-PMMS for goods if
\[
  v_i(X_i)
  \ge
  \alpha\,\GoodsMu_i(2,X_i\cup X_j)
  \qquad\text{for all }i\ne j.
\]
For a fixed allocation $X$, define its PMMS factor by
\[
  \alpha(X)
  :=
  \min_{\substack{i\ne j\\
  \GoodsMu_i(2,X_i\cup X_j)>0}}
  \frac{v_i(X_i)}
       {\GoodsMu_i(2,X_i\cup X_j)}.
\]
Hence, $X$ is $\alpha$-PMMS if and only if $\alpha(X)\ge \alpha$,
and $X$ is exact PMMS if and only if $\alpha(X)\ge 1$.

For a goods instance $I$, define its optimal PMMS factor by
\[
  \alpha^\star(I)
  :=
  \max_{X}\alpha(X),
\]
where the maximum ranges over all feasible allocations.
The universal PMMS guarantee for additive goods is
\[
  \alphaPMMS
  :=
  \inf_I \alpha^\star(I).
\]
Thus, $\alphaPMMS$ is the largest factor that can be guaranteed
for every additive-goods instance. In particular, exact PMMS
allocations exist universally if and only if $\alphaPMMS\ge 1$.

\paragraph{Chores.}
Each agent $i$ has a nonnegative additive cost function
$c_i:2^M\to\mathbb{R}_{\ge 0}$.
For $S\subseteq M$, define agent $i$'s two-way maximin-share benchmark
for chores by
\[
  \ChoreMu_i(2,S)
  :=
  \min_{(P_1,P_2)\in\PiTwo(S)}
  \max\{c_i(P_1),c_i(P_2)\}.
\]
Thus, agent $i$ partitions $S$ into two bundles and chooses a partition
that minimizes the cost of the more costly bundle.
An allocation $X=(X_1,\ldots,X_n)$ is $\rho$-PMMS for chores if
\[
  c_i(X_i)
  \le
  \rho\,\ChoreMu_i(2,X_i\cup X_j)
  \qquad\text{for all }i\ne j.
\]
For a fixed allocation $X$, define its PMMS factor by
\[
  \rho(X)
  :=
  \max_{\substack{i\ne j\\
  \ChoreMu_i(2,X_i\cup X_j)>0}}
  \frac{c_i(X_i)}
       {\ChoreMu_i(2,X_i\cup X_j)}.
\]
Hence $X$ is $\rho$-PMMS if and only if $\rho(X)\le \rho$,
and $X$ is exact PMMS if and only if $\rho(X)\le 1$.

For a chore instance $I$, define its optimal PMMS factor by
\[
  \rho^\star(I)
  :=
  \min_X \rho(X),
\]
where the minimum ranges over all feasible allocations.
The universal PMMS guarantee for additive chores is
\[
  \rhoPMMS
  :=
  \sup_I \rho^\star(I).
\]
Thus, $\rhoPMMS$ is the smallest factor that can be guaranteed
for every additive-chore instance. In particular, exact PMMS
allocations exist universally if and only if $\rhoPMMS\le 1$.

\paragraph{EFX for Chores.}
We use the positive-cost convention of Sun et al.~\cite{journals/aamas/SunCD23a}: an allocation $X$ is EFX if
\[
  c_i(X_i\setminus\{g\})\leq c_i(X_j)
  \qquad\text{for all }i\ne j\text{ and }g\in X_i\text{ with }c_i(g)>0.
\]
Thus, zero-cost chores are excluded from the removal condition. Under this convention, PMMS implies EFX for nonnegative additive costs, without any non-degeneracy assumption.

\section{Main Impossibility Results}
\label{sec:mainImpossibility}
In this section, we present the impossibility/hardness results for both settings with goods and chores.\footnote{
The exact verification scripts and outputs are available at
\url{https://anonymous.4open.science/r/pmms-verification-A58D}.
}
Under each setting, we show that there exists an instance where a PMMS allocation does not exist.
In addition, we show that the problem of deciding if a given instance admits a PMMS allocation is NP-hard for both goods and chores.

\subsection{Non-Existence Examples}
An additive chores instance where no EFX allocation exists is presented in \citet{journals/corr/abs-2606-08872}.
Since PMMS implies EFX, this immediately gives an example where no PMMS allocation exists for additive chores.
The non-existence example of PMMS for additive goods follows from the following lemma, which constructs an additive goods instance $I'$ from an additive chores instance $I$ such that $I$ admits a PMMS allocation if and only if $I'$ admits one.

\begin{lemma}[Karp Reduction from Chores to Goods]\label{lem:reduction_chores_goods}
    Given an additive chores instance $I=(N,M,c_1,\ldots,c_n)$ with $n\geq2$ agents and $m$ items, consider an additive goods instance $I'=(N,M',v_1,\ldots,v_n)$ with $n$ agents and $nm$ items constructed below:
    \begin{itemize}
        \item let $K=1+\max_ic_i(M)$;
        \item for each item $g$ in $I$, construct an item $g'$ in $I'$ such that $v_i(g')=K-c_i(g)$ for each $i\in N$;
        \item add $(n-1)m$ additional goods such that each agent has value $K$ for each good.
    \end{itemize}
    Then $I$ admits a PMMS allocation if and only if $I'$ admits a PMMS allocation.
\end{lemma}
\begin{proof}
The claim is immediate if $m=0$, so assume $m\geq1$.
Let $D$ denote the set of $(n-1)m$ additional goods.
For every $X\subseteq M'$, define $\pi(X):=\{g\in M:g'\in X\}$, i.e., $\pi(X)$ contains the original chores represented by the goods in $X$. 
By construction, $v_i(X)=K|X|-c_i(\pi(X))$.
Since costs are nonnegative and $K>c_i(M)$ for every $i\in N$, we have
\begin{equation}
    K|X|-c_i(M)\leq v_i(X)\leq K|X|.
    \label{eq:padding-bundle-value}
\end{equation}

We first show that every PMMS allocation of $I'$ gives exactly $m$ goods to each agent. 
Suppose otherwise, and let $A'=(A'_1,\ldots,A'_n)$ be such an allocation. 
Since $|M'|=nm$, there exist distinct agents $i,j$ such that
$$
    |A'_i|=r\leq m-1
    \qquad\text{and}\qquad
    |A'_j|\geq m+1\geq r+2.
$$
Consequently, $A'_i\cup A'_j$ can be partitioned into two bundles $P',Q'$ with $|P'|,|Q'|\geq r+1$. 
By
\eqref{eq:padding-bundle-value},
$$
    \min\{v_i(P'),v_i(Q')\}
    \geq (r+1)K-c_i(M)
    >rK
    \geq v_i(A'_i).
$$
This contradicts the PMMS condition for agent $i$ with respect to agent $j$. 
Therefore,
\begin{equation}
    |A'_i|=m
    \qquad\text{for every }i\in N
    \label{eq:padding-equal-cardinality}
\end{equation}
in every PMMS allocation of $I'$.

Next, consider any allocation $A'$ of $I'$ satisfying
\eqref{eq:padding-equal-cardinality}, and define $A_i:=\pi(A'_i)$ for each $i\in N$.
Then $A=(A_1,\ldots,A_n)$ is an allocation of $I$, and
\begin{equation}
    v_i(A'_i)=mK-c_i(A_i).
    \label{eq:padding-own-value}
\end{equation}
We establish an exact correspondence between the pairwise maximin shares in the two instances.

Fix distinct agents $i,j$. The set $A'_i\cup A'_j$ contains exactly $2m$ goods. Any bipartition of this set that is not $m$-versus-$m$ has a part containing at most $m-1$ goods, whose value to agent $i$ is at most $(m-1)K$. 
In contrast, every
$m$-versus-$m$ bipartition has both parts worth at least$mK-c_i(M)>(m-1)K$.
Hence, every bipartition attaining agent $i$'s two-agent maximin share of $A'_i\cup A'_j$ has exactly $m$ goods in each part.

Moreover, every bipartition $P\mathbin{\dot\cup}Q=A_i\cup A_j$
can be extended to an $m$-versus-$m$ bipartition of
$A'_i\cup A'_j$. 
Indeed, $|A_i\cup A_j|\leq m$, so both $|P|$ and $|Q|$ are at most $m$. 
The union $A'_i\cup A'_j$ contains exactly
$$
    2m-|A_i\cup A_j|
    =(m-|P|)+(m-|Q|)
$$
goods from $D$. Assigning $m-|P|$ of these goods to the copies of $P$ and the remaining $m-|Q|$ goods to the copies of $Q$ gives the required bipartition. Conversely, every $m$-versus-$m$ bipartition of $A'_i\cup A'_j$ induces a bipartition of $A_i\cup A_j$ by applying $\pi$ to its two parts.

Write $\mu_i^I(2,\cdot)$ and $\mu_i^{I'}(2,\cdot)$ for the two-agent maximin shares in $I$ and $I'$, respectively, using the min--max definition for chores and the max--min definition for goods. 
The preceding observations imply
\begin{align}
    \mu_i^{I'}(2,A'_i\cup A'_j)
    &=
    \max_{P\mathbin{\dot\cup}Q=A_i\cup A_j}
    \min\{mK-c_i(P),\,mK-c_i(Q)\}
    \notag\\
    &=
    mK-
    \min_{P\mathbin{\dot\cup}Q=A_i\cup A_j}
    \max\{c_i(P),\,c_i(Q)\}
    \notag\\
    &=
    mK-\mu_i^I(2,A_i\cup A_j).
    \label{eq:padding-share-identity}
\end{align}
Combining \eqref{eq:padding-own-value} and
\eqref{eq:padding-share-identity}, we obtain
\begin{align*}
    v_i(A'_i)\geq\mu_i^{I'}(2,A'_i\cup A'_j)
    &\iff
    mK-c_i(A_i)\geq mK-\mu_i^I(2,A_i\cup A_j)\\
    &\iff
    c_i(A_i)\leq\mu_i^I(2,A_i\cup A_j).
\end{align*}
Thus, for every equal-cardinality allocation $A'$ of $I'$, the allocation $A'$ is PMMS if and only if its decoded allocation $A$ is PMMS in $I$.

To conclude, suppose $I$ admits a PMMS allocation
$A=(A_1,\ldots,A_n)$. Give agent $i$ the goods
$\{g':g\in A_i\}$ together with $m-|A_i|$ goods from $D$.
This is feasible because
$$
    \sum_{i\in N}(m-|A_i|)
    =nm-m
    =|D|.
$$
The resulting allocation of $I'$ is equal-cardinality and hence is PMMS by the equivalence above.

Conversely, if $I'$ admits a PMMS allocation $A'$, then \eqref{eq:padding-equal-cardinality} holds. Decoding $A'$ therefore yields a PMMS allocation of $I$.
\end{proof}

\begin{remark}
Consider the two decision problems of deciding if a PMMS allocation exists for the goods and chores settings respectively.
Lemma~\ref{lem:reduction_chores_goods} gives a Karp reduction from the chores problem to the goods problem.
The lemma has established the completeness and the soundness of the reduction.
To see the reduction is polynomial-time, under the standard binary encoding of nonnegative rational singleton costs, $K$ and all constructed singleton
values have polynomial encoding length. 
The construction produces $nm$ goods and can be performed in polynomial time.
\end{remark}

We therefore have the following theorem. The nonexistence result for goods was also obtained concurrently and independently by \citet{aziz2026pmms} and \citet{golz2026pmms}.
\begin{theorem}\label{thm:goods}
    There is an additive chores instance where no PMMS allocation exists. The same holds for additive goods.
\end{theorem}
\begin{proof}
    The first part of the theorem follows from the existence of chore instances with no EFX allocation due to~\cite{journals/corr/abs-2606-08872}, and the fact that PMMS implies EFX.
    Given an additive chores instance with no PMMS allocation, we can apply Lemma~\ref{lem:reduction_chores_goods} to construct an additive goods instance with no PMMS allocation, which proves the second part.
\end{proof}

In Appendix~\ref{append:impossibility_ratio}, we extend Theorem~\ref{thm:goods} to its gapped version.
For additive goods, we show that there exists an instance where the approximation ratio of PMMS is at most $\frac{226}{227}$.
The concurrent work of \citet{golz2026pmms} gives a stronger impossibility bound of $\frac{78}{79}$.
For additive chores, we show that there exists an instance where the approximation ratio is at least $\frac{220413}{200000}=1.102065$.

\subsection{NP-hardness}
In this section, we prove the following result.
\begin{theorem}\label{thm:NPhard-chores}
    Given an additive chores instance, it is NP-hard to decide if the instance admits a PMMS allocation.
\end{theorem}

Theorem~\ref{thm:NPhard-chores} implies the following theorem due to the Karp-reduction in Lemma~\ref{lem:reduction_chores_goods}.

\begin{theorem}\label{thm:NPhard-goods}
    Given an additive goods instance, it is NP-hard to decide if the instance admits a PMMS allocation.
\end{theorem}

%is % TODO: Supply the proof of Theorem~\ref{thm:NPhard-chores} before treating NP-hardness as an established contribution.
%The statements above assert NP-hardness, without claiming membership in NP. An allocation alone does not provide an immediate polynomial-time certificate, since verifying PMMS requires checking all bipartitions of every pairwise union.
As a remark, the NP-hardness in the two theorems above cannot be replaced by NP-completeness, as verifying if an allocation satisfies PMMS is not known to be polynomial-time solvable (in fact, this decision problem is coNP-complete).

We will prove Theorem~\ref{thm:NPhard-chores} in the remaining part of this section.
Our proof makes use of the result of~\citet{hosseini2023fairly}, which shows the NP-completeness of deciding if there is an EFX allocation in the setting with both goods and chores where the valuations are lexicographical.
We will first introduce the setting of mixed goods and chores with lexicographical valuations, and review the NP-completeness result of~\citet{hosseini2023fairly}.
Then, we will prove Theorem~\ref{thm:NPhard-chores} by a reduction to this NP-completeness result.

\subsubsection{Review of Hosseini et al.'s NP-completeness Result}
We first define the setting with mixed goods and chores.
Let $N$ be the set of agents and $M$ be the set of items as before.
Each agent $i$ has an additive valuation function $u_i:2^M\to\mathbb{R}$.
Notice that in the mixed setting, $u_i(g)$ (for an item $g\in M$) can be both non-negative and negative.
An item $g$ with $u_i(g)\geq0$ is a good for agent $i$; otherwise, it is a chore for agent $i$.
Let $G_i$ and $C_i$ be the sets of goods and chores for agent $i$, respectively.

We then define \emph{lexicographical valuations}.
Instead of using the original definition in~\citet{hosseini2023fairly}, we use the following equivalent definition for the purpose of this paper.
In the setting with additive mixed goods and chores, a valuation function $u$ is \emph{lexicographical} if items can be order $g_1,\ldots,g_m$ such that $u(g_i)\in\{2^{m-i},-2^{m-i}\}$.
Intuitively, an item ranked $i$-th by $u$ is more significant than the combination of all the items ranked after the $i$-th: $|u(g_i)|>\sum_{j>i}|u(g_j)|$.

We next extend the definition of EFX to this mixed setting.
Given an allocation $X=(X_1,\ldots,X_n)$, we say that $i$ envies $j$ if $u_i(X_i)<u_i(X_j)$.
We say that $X$ is EFX if for every pair of agents $i,j$ such that $i$ envies $j$, it holds that
\begin{itemize}
    \item for each $g\in G_i\cap X_j$, we have $u_i(X_i)\geq u_i(X_j\setminus\{g\})$, and
    \item for each $g\in C_i\cap X_i$, we have $u_i(X_i\setminus\{g\})\geq u_i(X_j)$.
\end{itemize}

The following result is due to~\citet{hosseini2023fairly}.

\begin{theorem}[\citet{hosseini2023fairly}]\label{thm:hadi}
    Given an instance with mixed goods and chores where valuations are lexicographical, it is NP-complete to decide if an EFX allocation exists.
\end{theorem}

Finally, we define the notion of PMMS for this mixed setting.
For $S\subseteq M$ and each agent $i$, define
$$\mu_i(2,S):=\max_{(P_1,P_2)\in\Pi_2(S)}\min\{u_i(P_1),u_i(P_2)\}$$
as in the setting with goods.
An allocation $X=(X_1,\ldots,X_n)$ is PMMS if
$$u_i(X_i)\geq\mu_i(2,X_i\cup X_j)\qquad\mbox{for all }i\neq j.$$
Notice that this definition is the same as the one for goods, except that $u_i$ can be negative now.

\subsubsection{Proof of Theorem~\ref{thm:NPhard-chores}}
We first show the following lemma, which shows that, under lexicographical valuations, the notions of EFX and PMMS are equivalent.
This lemma may be of independent interest.
After this, we will prove Theorem~\ref{thm:NPhard-chores} by a reduction from the mixed instance to the chores-only instance that uses the same ideas in Lemma~\ref{lem:reduction_chores_goods}.

\begin{lemma}\label{lem:EFX-PMMS}
    Given an instance with mixed goods and chores and lexicographical valuations, an allocation $X$ is EFX if and only if it is PMMS.
\end{lemma}
\begin{proof}
Notice that, in our definition of lexicographical valuations, every good has strictly positive value and every chore has strictly negative value.

\paragraph{PMMS implies EFX.}
Suppose that $X$ is not EFX. Then there exist distinct agents $i,j$ for whom one of the following violations occurs.

First, suppose there is a good $g\in G_i\cap X_j$ such that $u_i(X_i)<u_i(X_j\setminus\{g\})$.
Since $u_i(g)>0$, the bipartition
$$
    \bigl(X_i\cup\{g\},\,X_j\setminus\{g\}\bigr)
    \in\Pi_2(X_i\cup X_j)
$$
has both parts strictly more valuable to agent $i$ than $X_i$.
Consequently, $\mu_i(2,X_i\cup X_j)>u_i(X_i)$, violating PMMS.

Second, suppose there is a chore $g\in C_i\cap X_i$ such that $u_i(X_i\setminus\{g\})<u_i(X_j)$.
Since $u_i(g)<0$, we have $u_i(X_i\setminus\{g\})>u_i(X_i)$,
and the assumed EFX violation also gives $u_i(X_j\cup\{g\})=u_i(X_j)+u_i(g)>u_i(X_i)$.
Thus, the bipartition
$$
    \bigl(X_i\setminus\{g\},\,X_j\cup\{g\}\bigr)
    \in\Pi_2(X_i\cup X_j)
$$
again has both parts strictly more valuable to agent $i$ than $X_i$, contradicting PMMS. Therefore, every PMMS allocation is EFX.

\paragraph{EFX implies PMMS.}
Suppose that $X$ is EFX, and fix distinct agents $i,j$.
Let
\[
    S:=X_i\cup X_j,
    \qquad
    G:=G_i\cap S,
    \qquad
    C:=C_i\cap S.
\]
If $u_i(X_i)\geq u_i(X_j)$, additivity immediately gives
\[
    \mu_i(2,S)
    \leq \frac{u_i(S)}{2}
    =\frac{u_i(X_i)+u_i(X_j)}{2}
    \leq u_i(X_i).
\]
Hence, assume that $u_i(X_i)<u_i(X_j)$.

Let $e$ be the most important item in $S$ according to
agent $i$. Lexicographical valuations satisfy
\[
    |u_i(e)|
    >
    \sum_{g\in S\setminus\{e\}}|u_i(g)|.
\]
Thus, in any bipartition of $S$, the part containing $e$ is strictly preferred if $e$ is a good and strictly dispreferred if $e$ is a chore. Since agent $i$ envies agent $j$, there are only two cases.

\medskip
\noindent\emph{Case 1: $e\in G_i\cap X_j$.}
The bundle $X_i$ cannot contain a chore: removing such a chore from $X_i$ would leave the dominant good $e$ in $X_j$, so agent $i$ would still envy agent $j$, contrary to EFX. 
Similarly, $X_j$ cannot contain another good
besides $e$: removing that other good would again leave $e$ in $X_j$, and envy would persist. Therefore,
\[
    X_i=G\setminus\{e\},
    \qquad
    X_j=C\cup\{e\}.
\]

Now consider any bipartition $(P_1,P_2)\in\Pi_2(S)$.
One of its parts, say $P_1$, does not contain $e$.
Its value is at most the total value of all goods in
$S\setminus\{e\}$, so
\[
    \min\{u_i(P_1),u_i(P_2)\}
    \leq u_i(P_1)
    \leq u_i(G\setminus\{e\})
    =u_i(X_i).
\]
Taking the maximum over all bipartitions yields
$\mu_i(2,S)\leq u_i(X_i)$.

\medskip
\noindent\emph{Case 2: $e\in C_i\cap X_i$.}
The bundle $X_i$ cannot contain another chore besides
$e$: removing that other chore would leave the dominant chore $e$ in $X_i$, so agent $i$ would still envy agent $j$. 
Likewise, $X_j$ cannot contain a good: removing
that good would leave $e$ in $X_i$, and envy would
persist. 
Both possibilities contradict EFX. 
Hence,
\[
    X_i=G\cup\{e\},
    \qquad
    X_j=C\setminus\{e\}.
\]

For any bipartition $(P_1,P_2)\in\Pi_2(S)$, one part,
say $P_1$, contains $e$. 
Its value is at most the value of $e$ together with all goods in $S$. 
Thus,
\[
    \min\{u_i(P_1),u_i(P_2)\}
    \leq u_i(P_1)
    \leq u_i(e)+u_i(G)
    =u_i(X_i).
\]
Again, maximizing over all bipartitions gives
$\mu_i(2,S)\leq u_i(X_i)$.

In every case, agent $i$ satisfies PMMS with respect
to agent $j$. Since $i,j$ were arbitrary, $X$ is PMMS.
\end{proof}

Now, we are ready to prove Theorem~\ref{thm:NPhard-chores}.
Given a mixed lexicographical instance $I=(N,M,u_1,\ldots,u_n)$, we construct a chores-only instance $I'=(N,M',c_1,\ldots,c_n')$ as follows.
Let $K=2^{2m+1}$.
For each $g\in M$, construct a chore $g'\in M'$ such that $c_i(g')=K-u_i(g)$ for each $i\in N$.
After this, add $(n-1)m$ additional chores such that each agent has cost $K$ for each of them.
Based on Lemma~\ref{lem:EFX-PMMS}, we can conclude Theorem~\ref{thm:NPhard-chores} if we can show that $I$ admits a PMMS allocation if and only if $I'$ does.
We establish this equivalence by an analogous padding
argument to that in Lemma~\ref{lem:reduction_chores_goods},
with the max--min and min--max benchmarks interchanged.

Let
\[
    T:=\max_{i\in N}\sum_{g\in M}|u_i(g)|=2^m-1.
\]
Our choice of $K$ satisfies $K>2T$. For every
$Y\subseteq M'$, define
\[
    \pi(Y):=\{g\in M:g'\in Y\}.
\]
Then
\[
    c_i(Y)=K|Y|-u_i(\pi(Y)),
    \qquad
    K|Y|-T\leq c_i(Y)\leq K|Y|+T.
\]

First, every PMMS allocation $X'$ of $I'$ gives exactly $m$ chores to each agent. 
Otherwise, some agent $i$ receives $r\geq m+1$ chores and another agent $j$ receives at most $m-1\leq r-2$ chores. 
Their combined bundle can
be partitioned into two parts, each containing at most $r-1$ chores. 
Each part then costs agent $i$ at most
\[
    (r-1)K+T<rK-T\leq c_i(X'_i),
\]
contradicting PMMS.

Now consider any allocation $X'$ giving exactly $m$
chores to each agent, and let $X_i:=\pi(X'_i)$.
For each pair $i\neq j$, every min--max-optimal
bipartition of $X'_i\cup X'_j$ is $m$-versus-$m$:
every such equal-cardinality bipartition has maximum
cost at most $mK+T$, whereas an unequal-cardinality
bipartition has a part costing at least
\[
    (m+1)K-T>mK+T.
\]
As in Lemma~\ref{lem:reduction_chores_goods}, every
bipartition of $X_i\cup X_j$ can be completed to an
$m$-versus-$m$ bipartition of $X'_i\cup X'_j$ using
the additional chores.

Writing $\mu_i^I$ for the max--min benchmark in the
mixed instance and $\mu_i^{I'}$ for the min--max
benchmark in the chore instance, we therefore obtain
\begin{align*}
    \mu_i^{I'}(2,X'_i\cup X'_j)
    &=
    \min_{(P_1,P_2)\in\Pi_2(X_i\cup X_j)}
    \max\{mK-u_i(P_1),\,mK-u_i(P_2)\}\\
    &=
    mK-\mu_i^I(2,X_i\cup X_j).
\end{align*}
Since $c_i(X'_i)=mK-u_i(X_i)$, it follows that
\[
    c_i(X'_i)\leq\mu_i^{I'}(2,X'_i\cup X'_j)
    \quad\Longleftrightarrow\quad
    u_i(X_i)\geq\mu_i^I(2,X_i\cup X_j).
\]
Thus, every PMMS allocation of $I'$ decodes to a PMMS
allocation of $I$. 
Conversely, any PMMS allocation $X$ of $I$ can be extended to an allocation of $I'$ by giving agent $i$ the copies of the items in $X_i$ and $m-|X_i|$ additional chores. 
The resulting allocation is PMMS by the same equivalence.

Finally, the construction produces $nm$ chores,
and every singleton cost is a positive integer
with $O(m)$ bits. 
Hence, the construction is polynomial-time. 
Together with Theorem~\ref{thm:hadi} and Lemma~\ref{lem:EFX-PMMS}, this proves Theorem~\ref{thm:NPhard-chores}.

\section{Existence of \texorpdfstring{$\frac43$}{4/3}-PMMS Allocations for Chores}
\label{sec:chores-existence}

In this section, we establish the existence of approximate-PMMS allocations for additive chores. Our proof starts from a pEF1 allocation together with supporting prices, obtained from the existence proof of \citet{conf/soda/Mahara26}. We then transform this allocation through a reassignment of selected chores followed by a single pass of local exchanges. The key objective is to obtain an allocation in which every agent either receives a singleton bundle or incurs cost at most twice the cost she assigns to any other bundle. We show that this structural property is sufficient to guarantee a $\frac43$-approximation of PMMS. We first carry out the construction for instances with strictly positive costs that are non-degenerate, and then extend the result to general additive-chore instances by a perturbation and continuity argument.

\begin{theorem}\label{thm:chores-existence}
Every additive-chore instance admits a $\frac43$-PMMS allocation.
Consequently,
\[
    \rhoPMMS \leq \frac43.
\]
\end{theorem}

\subsection{Proof Approach and Basic Observations}
Before presenting our construction, we first discuss the price-based ingredient from prior work~\cite{conf/sigecom/BarmanKV18,conf/soda/Mahara26} and explain why an additional transformation is needed for PMMS.

\paragraph{Price-EF1 as a starting point.}
Price envy-freeness up to one item (pEF1) was introduced by \citet{conf/sigecom/BarmanKV18} and has subsequently been used in market-based approaches for both goods and chores; see, e.g., \citet{conf/soda/Mahara26}. Let $p=(p_g)_{g\in M}$ be a positive price vector, where $p_g$ denotes the price of chore $g$. For any bundle $S\subseteq M$, let $p(S):=\sum_{g\in S}p_g$. For a nonempty bundle $S$, define $\hat p(S):=p(S)-\max_{g\in S}p_g$, and set $\hat p(\emptyset):=0$. Thus, $\hat p(S)$ is the price of $S$ after removing a maximum-price chore. An allocation $Y=(Y_i)_{i\in N}$ is pEF1 if, for every pair of agents $i,j\in N$, either $Y_i=\emptyset$, or there exists a chore $g\in Y_i$ such that $p(Y_i\setminus\{g\})\leq p(Y_j)$. Equivalently, as observed by \citet{conf/soda/Mahara26}, $Y$ is pEF1 if $\max_{i\in N}\hat p(Y_i)\leq\min_{j\in N}p(Y_j)$.

For our proof, we use the pEF1 allocation obtained in the existence proof of \citet{conf/soda/Mahara26}, together with its supporting prices. After independently rescaling the agents' cost functions by positive factors, these prices satisfy $p_g\leq c_i(g)$ for every $i\in N$ and $g\in M$, with equality whenever $g\in Y_i$. Hence, $p(S)\leq c_i(S)$ for every agent $i$ and bundle $S\subseteq M$, while $p(S)=c_i(S)$ whenever $S\subseteq Y_i$. More importantly, combining these supporting-price relations with pEF1 yields two complementary properties: after removing a maximum-price chore from an allocated bundle, the remaining price is no larger than the minimum price of an allocated bundle, while every allocated bundle has cost at least this common price threshold to every agent. These properties allow us to convert the pEF1 structure into the cost comparisons used in our transformation.

\paragraph{Why pEF1 is not sufficient for PMMS.}
A supported pEF1 allocation need not itself be $\frac43$-PMMS. Consider two identical agents and three chores $g_1,g_2,g_3$, as shown in Table~\ref{tab:pef1-counterexample}.

\begin{table}[!ht]
\centering
\caption{A supported pEF1 allocation that is not $\frac43$-PMMS.}
\label{tab:pef1-counterexample}
\begin{tabular}{c|c|c|c|c}
\hline
Chore & $c_1(g)$ & $c_2(g)$ & $p_g$ & Owner \\
\hline
$g_1$ & $2$ & $2$ & $2$ & Agent $1$ \\
$g_2$ & $1$ & $1$ & $1$ & Agent $1$ \\
$g_3$ & $1$ & $1$ & $1$ & Agent $2$ \\
\hline
\end{tabular}
\end{table}

Let $Y_1=\{g_1,g_2\}$ and $Y_2=\{g_3\}$. Since the prices coincide with the costs of both agents, the supporting-price conditions are satisfied. Moreover, $\hat p(Y_1)=1$ and $\hat p(Y_2)=0$, while $\min\{p(Y_1),p(Y_2)\}=1$. Hence, $Y$ is pEF1.
However, agent $1$ incurs cost $c_1(Y_1)=3$. To evaluate $\ChoreMu_1(2,Y_1\cup Y_2)$ for $Y_1\cup Y_2=\{g_1,g_2,g_3\}$, consider the partition $(P_1,P_2)=(\{g_1\},\{g_2,g_3\})\in\PiTwo(Y_1\cup Y_2)$. Both parts have cost $2$, so $\ChoreMu_1(2,Y_1\cup Y_2)\leq 2$. On the other hand, the total cost of $Y_1\cup Y_2$ is $4$, so every partition in $\PiTwo(Y_1\cup Y_2)$ has a part of cost at least $2$. Hence, $\ChoreMu_1(2,Y_1\cup Y_2)=2$. Thus, the pairwise PMMS ratio of agent $1$ against agent $2$ is $c_1(Y_1)/\ChoreMu_1(2,Y_1\cup Y_2)=3/2>4/3$.
This example highlights the limitation of pEF1 for our purpose. The pEF1 condition controls an agent's bundle only after one high-price chore is removed, whereas PMMS compares the cost of the entire bundle with the two-way maximin-share benchmark of the union. Thus, a supported pEF1 allocation provides useful structure, but an additional transformation is needed to obtain the $\frac43$-PMMS guarantee.

\paragraph{High-level idea.}
Our goal is to transform the supported pEF1 allocation into an allocation $X=(X_i)_{i\in N}$ with the following structural property: \emph{for every agent $i$, either $X_i$ is a singleton, or $c_i(X_i)\leq 2c_i(X_j)$ for every $j\neq i$}. As shown below, either condition is sufficient to guarantee $\frac43$-PMMS.

To achieve this property, we first identify a set of high-price chores and reassign them according to a fixed order. We then process the corresponding high agents in the same order. When an agent's current bundle costs more than twice her least-costly bundle among the other agents, we perform a local exchange. The ordered reassignment ensures that singletons created by later exchanges remain sufficiently costly to agents processed earlier. Thus, for each processed high agent, later exchanges either preserve the established pairwise cost bound or assign her a singleton, which satisfies exact PMMS.

We first recall two standard lower bounds on the maximin-share benchmark for chores. The following is the specialization of \citet[Lemma~2.1]{journals/aamas/SunCD23a} to two-way partitions.

\begin{lemma}[Basic lower bounds]
\label{lem:basic-bound}
For every agent $i\in N$ and every bundle $S\subseteq M$, we have $\ChoreMu_i(2,S)\geq c_i(S)/2$. Moreover, for every chore $g\in S$, we have $\ChoreMu_i(2,S)\geq c_i(g)$.
\end{lemma}
The first bound yields the following relation, which explains the factor $2$ used in our exchange rule.
\begin{lemma}[Two-load bound]
\label{lem:2-load-bound}
Let $A,B\subseteq M$ be disjoint, and let $x:=c_i(A)$ and $y:=c_i(B)$. If $x\leq 2y$, then $x\leq \frac43\ChoreMu_i(2,A\cup B)$. Equivalently, if $x>\frac43\ChoreMu_i(2,A\cup B)$, then $x>2y$.
\end{lemma}

\begin{proof}
Suppose $x\leq 2y$. Then $3x\leq 2(x+y)$, and hence $x\leq 2(x+y)/3$. By Lemma~\ref{lem:basic-bound} and additivity, $\ChoreMu_i(2,A\cup B)\geq (x+y)/2$. Therefore, $x\leq 2(x+y)/3\leq \frac43\ChoreMu_i(2,A\cup B)$. The second statement follows from the strict contrapositive.
\end{proof}
Note that the converse of Lemma~\ref{lem:2-load-bound} is not required: $x>2y$ does not necessarily imply a $\frac43$-PMMS violation, and we use this inequality only as the condition for performing an exchange. On the other hand, the second bound in Lemma~\ref{lem:basic-bound} implies that singleton bundles satisfy exact PMMS: if $A=\{g\}$, then for every bundle $B$ disjoint from $A$, we have $c_i(A)=c_i(g)\leq\ChoreMu_i(2,A\cup B)$. Thus, it remains only to construct an allocation in which every agent either receives a singleton or satisfies $c_i(X_i)\leq 2c_i(X_j)$ for every $j\neq i$.

\subsection{A pEF1 Starting Allocation}

We first consider instances with strictly positive costs that are non-degenerate in the sense of \citet[full version, Section~2]{conf/soda/Mahara26}. Specifically, for every cycle $C=(i_1,g_1,i_2,g_2,\ldots,i_k,g_k,i_{k+1})$ in the complete agent--chore bipartite graph, where $i_{k+1}=i_1$, we have $\prod_{\ell=1}^k c_{i_\ell}(g_\ell)\neq\prod_{\ell=1}^k c_{i_{\ell+1}}(g_\ell)$. The positivity and non-degeneracy assumptions will be removed at the end of the section.

We use the following consequence of \citet[full version, Lemmas~3.4 and~3.5]{conf/soda/Mahara26}.
\begin{theorem}
\label{thm:scale}
Let the instance have strictly positive costs and be non-degenerate. Then, there exist positive weights $(w_i)_{i\in N}$, an allocation $Y=(Y_i)_{i\in N}$, and a positive price vector $p=(p_g)_{g\in M}$ such that, under the rescaled costs $\widetilde c_i=w_i c_i$, we have $p_g\leq \widetilde c_i(g)$ for every $i\in N$ and $g\in M$, with equality whenever $g\in Y_i$. Moreover, $Y$ is pEF1 with respect to $p$, i.e., $\max_{i\in N}\hat p(Y_i)\leq\min_{j\in N}p(Y_j)$.
\end{theorem}
The rescaling does not affect PMMS, since multiplying $c_i$ by $w_i>0$ multiplies both $c_i(X_i)$ and $\ChoreMu_i(2,X_i\cup X_j)$ by the same factor. We therefore suppress the tildes and write $c_i$ for the rescaled costs. By Theorem~\ref{thm:scale}, $p(S)\leq c_i(S)$ for every $i\in N$ and $S\subseteq M$, with equality whenever $S\subseteq Y_i$.

Assume $n\geq2$ and $m>n$, and let $\tau:=\min_{j\in N}p(Y_j)$. Every bundle $Y_i$ is nonempty. Otherwise, $\tau=0$, and pEF1 would imply $\hat p(Y_i)=0$ for every $i\in N$. Since all prices are positive, each nonempty bundle would then contain at most one chore, contradicting $m>n$. Hence, $\tau>0$.
For each agent $i$, choose a maximum-price chore $h_i'\in Y_i$ and let $S_i:=Y_i\setminus\{h_i'\}$. Following the high--low decomposition used by \citet{conf/aaai/GargM26}, define $N_0:=\{i\in N:p_{h_i'}\leq\tau\}$ and $N_H:=\{i\in N:p_{h_i'}>\tau\}$. We call the agents in $N_0$ and $N_H$ low and high agents, respectively, and let $H:=\{h_i':i\in N_H\}$.
The following lemma collects the bounds needed in the subsequent transformation.
\begin{lemma}[Price bounds]
\label{lem:price-bounds}
For every $i\in N$, we have $c_i(S_i)=p(S_i)\leq\tau$. Moreover, $c_i(Y_i)\leq2\tau$ for every $i\in N_0$, $c_i(Y_j)\geq\tau$ for every $i,j\in N$, and $c_i(g)>\tau$ for every $i\in N$ and $g\in H$.
\end{lemma}

\begin{proof}
Since $S_i\subseteq Y_i$, we have $c_i(S_i)=p(S_i)$. By the choice of $h_i'$ and pEF1, $p(S_i)=\hat p(Y_i)\leq\tau$. If $i\in N_0$, then $c_i(Y_i)=p(Y_i)=p(S_i)+p_{h_i'}\leq2\tau$. For every $i,j\in N$, we have $c_i(Y_j)\geq p(Y_j)\geq\tau$. Finally, if $g\in H$, then $p_g>\tau$, and hence $c_i(g)\geq p_g>\tau$ for every $i\in N$.
\end{proof}

Thus, low agents have cost at most $2\tau$, while every allocated bundle costs at least $\tau$ to every agent. For high agents, the residual bundles $S_i$ have cost at most $\tau$, whereas every chore in $H$ costs more than $\tau$ to every agent. These bounds will be used in the reassignment and exchange phases.

\subsection{The Allocation Algorithm}

We now transform the pEF1 allocation $Y$ constructed in the previous subsection; see Algorithm~\ref{alg:main}. Fix an arbitrary order of the high agents and, for notational convenience, relabel them as $N_H=\{1,\ldots,r\}$, where $r=|N_H|$. We use the same order in both phases.

\paragraph{Phase 1: Reassign the high-price chores.}
For each high agent $i$, keep the residual bundle $S_i$ unchanged and reassign only the distinguished chores in $H$. Processing the high agents in the order $1,\ldots,r$, let each agent $i$ select a minimum-cost chore among those in $H$ that remain unassigned; denote her choice by $h_i$. Since each chore is selected exactly once, $(h_i)_{i\in N_H}$ is a permutation of $H$. We then assign $S_i\cup\{h_i\}$ to each high agent $i$, while every low agent keeps her original bundle $Y_i$. Moreover, for any $i<k$, the chore $h_k$ was still available when agent $i$ made her choice, and hence $c_i(h_i)\leq c_i(h_k)$.

\paragraph{Phase 2: Perform one pass of exchanges.}
Process the high agents in the same order as in Phase~1. When agent $i$ is processed, choose an agent $\ell\neq i$ whose current bundle minimizes $c_i(X_\ell)$. If $c_i(X_i)\leq 2c_i(X_\ell)$, leave the allocation unchanged. Otherwise, perform the simultaneous exchange
$(X_i,X_\ell)\leftarrow(S_i\cup X_\ell,\{h_i\})$.
Thus, agent $i$ keeps her residual bundle $S_i$ and receives her least-costly bundle among the other agents, while agent $\ell$ receives the single chore $h_i$.

The ordered reassignment in Phase~1 and the local exchange in Phase~2 adapt the reassignment-and-swap framework of \citet{conf/aaai/GargM26}; related chore-exchange operations also appear in \citet{conf/stoc/GargMQ25}. The threshold $c_i(X_i)>2c_i(X_\ell)$ is chosen according to Lemma~\ref{lem:2-load-bound}, so that the transformation can establish the pairwise cost bound needed for $\frac43$-PMMS. Note also that all ties are broken arbitrarily. Throughout Phase~2, the residual bundles $S_i$ and the selected chores $h_i$ remain fixed; only the current allocation $X$ changes.

\begin{algorithm}[!ht]
\caption{Construction of a $\frac43$-PMMS Allocation}
\label{alg:main}
\DontPrintSemicolon

\KwIn{The allocation $Y$ and the sets $(S_i)_{i\in N}$, $N_H$, $H$, and $N_0$ defined above}
\KwOut{An allocation $X$}

Fix an arbitrary order of $N_H$ and relabel its agents as $1,\ldots,r$\;

$R\gets H$\tcp*{Unassigned high-price chores}

\tcp{Phase 1: Reassign the high-price chores}
\For{$i=1,\ldots,r$}{
    Choose $h_i\in\arg\min_{g\in R} c_i(g)$\;
    $R\gets R\setminus\{h_i\}$\;
}

$X_i\gets Y_i$ for every $i\in N_0$\;

\For{$i=1,\ldots,r$}{
    $X_i\gets S_i\cup\{h_i\}$\;
}

\tcp{Phase 2: Perform one pass of exchanges}
\For{$i=1,\ldots,r$}{
    Choose $\ell\in\arg\min_{k\in N\setminus\{i\}} c_i(X_k)$\;
    \If{$c_i(X_i)>2c_i(X_\ell)$}{
        $B\gets X_\ell$\;
        $(X_i,X_\ell)\gets(S_i\cup B,\{h_i\})$\;
    }
}

\KwRet{$X$}\;
\end{algorithm}

\subsection{Correctness of the Algorithm}

We now prove that Algorithm~\ref{alg:main} returns an allocation $X$ in which every agent $i$ either receives a singleton or satisfies $c_i(X_i)\leq 2c_i(X_j)$ for every $j\neq i$. By Lemma~\ref{lem:2-load-bound} and the second bound in Lemma~\ref{lem:basic-bound}, this is sufficient to guarantee $\frac43$-PMMS.
Let $X^0$ denote the allocation after Phase~1, and let $X^t$ denote the allocation after high agent $t$ is processed in Phase~2. For every $i\in N_H$, define $s_i:=c_i(S_i)$ and $a_i:=c_i(h_i)$. Since $h_i\in H$, Lemma~\ref{lem:price-bounds} gives $s_i\leq\tau<a_i$. In particular, $(s_i+a_i)/2<a_i$.

We first record the ordering property established in Phase~1.

\begin{lemma}[Ordered choices]
\label{lem:ordered-choices}
For any $i,k\in N_H$ with $i<k$, we have $c_i(h_i)\leq c_i(h_k)$.
\end{lemma}

\begin{proof}
The chore $h_k$ was still available when agent $i$ selected $h_i$. Since $i$ chose a minimum-cost remaining chore, $c_i(h_i)\leq c_i(h_k)$.
\end{proof}

The next lemma captures the main invariant of Phase~2. In particular, it guarantees that an unprocessed high agent still holds her Phase~1 bundle when her turn arrives. The ordering argument follows the same idea used to protect unprocessed agents in the reassignment framework of \citet{conf/aaai/GargM26}.

\begin{lemma}[Phase~2 invariant]
\label{lem:phase2-invariant}
For every $t\in\{0,\ldots,r\}$, the following properties hold after turn $t$:
\begin{enumerate}
    \item $X^t$ is a complete allocation and $p(X_j^t)\geq\tau$ for every $j\in N$;
    \item every unprocessed high agent $k>t$ holds $X_k^t=S_k\cup\{h_k\}$;
    \item every agent that has been chosen as $\ell$ in an exchange by time $t$ holds a singleton.
\end{enumerate}
Consequently, $c_i(X_j^t)\geq\tau$ for every $i,j\in N$.
\end{lemma}

\begin{proof}
We proceed by induction on $t$. For $t=0$, Phase~1 only permutes the chores in $H$ among the high agents, while every residual bundle $S_i$ remains with its original owner. Hence, $X^0$ is a complete allocation. If $j\in N_0$, then $X_j^0=Y_j$ and $p(X_j^0)=p(Y_j)\geq\tau$. If $j\in N_H$, then $X_j^0=S_j\cup\{h_j\}$ with $h_j\in H$, and therefore $p(X_j^0)\geq p_{h_j}>\tau$. The second property follows from the construction of Phase~1, and the third holds trivially.

Suppose the claim holds after turn $t-1$. By the second property, agent $t$ holds $S_t\cup\{h_t\}$ when her turn begins, and hence $c_t(X_t^{t-1})=s_t+a_t$. If no exchange occurs, the allocation remains unchanged. Suppose instead that agent $t$ exchanges with some agent $\ell$.

We first show that $\ell$ cannot be an unprocessed high agent. For every $k>t$, the induction hypothesis gives $X_k^{t-1}=S_k\cup\{h_k\}$. Hence,
$c_t(X_k^{t-1})\geq c_t(h_k)\geq c_t(h_t)=a_t>(s_t+a_t)/2$,
where the second inequality follows from Lemma~\ref{lem:ordered-choices}. On the other hand, since an exchange occurs at turn $t$, the selected agent $\ell$ satisfies $c_t(X_\ell^{t-1})<(s_t+a_t)/2$. Therefore, $\ell$ cannot be an unprocessed high agent.

Let $B=X_\ell^{t-1}$. The exchange replaces the two bundles $S_t\cup\{h_t\}$ and $B$ by $S_t\cup B$ and $\{h_t\}$. Since the two original bundles are disjoint, the two new bundles are also disjoint and have the same union. Thus, completeness is preserved. Moreover,
$p(S_t\cup B)=p(S_t)+p(B)\geq p(B)\geq\tau$,
while $p(\{h_t\})=p_{h_t}>\tau$. Hence, the first property remains valid. Since no unprocessed high agent can be chosen as $\ell$, the second property is also preserved.

It remains to verify the third property. Agent $\ell$ receives the singleton $\{h_t\}$. If $\ell\in N_H$, then her turn has already occurred, since an unprocessed high agent cannot be chosen as $\ell$. Moreover, agent $t$ could not have been chosen as $\ell$ earlier, for the same reason. Hence, every agent chosen as $\ell$ before turn $t$ either keeps her singleton or, if chosen again at turn $t$, receives another singleton. Therefore, every agent that has been chosen as $\ell$ by time $t$ holds a singleton.

The final claim follows from the supporting-price inequality: $c_i(X_j^t)\geq p(X_j^t)\geq\tau$ for all $i,j\in N$.
\end{proof}

The Phase~2 invariant describes what happens to unprocessed high agents. We also need to show that, once a processed high agent $i$ has established a lower bound on $c_i(X_j)$ for every $j\neq i$, later exchanges preserve this bound as long as agent $i$ is never chosen as $\ell$. The next lemma formalizes this observation.

\begin{lemma}[Persistence of lower bounds]
\label{lem:persistent-lower-bound}
Fix $i\in N_H$ and let $b\leq a_i$. Suppose that immediately after agent $i$ is processed, $c_i(X_j^i)\geq b$ for every $j\neq i$. If agent $i$ is never chosen as $\ell$ in a later exchange, then she keeps $X_i^i$, and in the final allocation, $c_i(X_j)\geq b$ for every $j\neq i$.
\end{lemma}

\begin{proof}
We prove the claim by induction over the turns following agent $i$. Immediately after agent $i$ is processed, the lower bound $c_i(X_j^i)\geq b$ holds for every $j\neq i$ by assumption.
Consider a later turn $k>i$, and suppose that immediately before this turn, $c_i(X_j)\geq b$ for every $j\neq i$. If no exchange occurs, the allocation does not change. Otherwise, let $\ell$ be the agent chosen for the exchange and let $B=X_\ell$ denote her bundle before the exchange. Since agent $i$ is never chosen as $\ell$, we have $\ell\neq i$, and hence $c_i(B)\geq b$ by the induction hypothesis.
After the exchange, the two affected bundles are $S_k\cup B$ and $\{h_k\}$. We have
$c_i(S_k\cup B)\geq c_i(B)\geq b$
and
$c_i(h_k)\geq c_i(h_i)=a_i\geq b$,
where the first inequality in the second relation follows from Lemma~\ref{lem:ordered-choices}. All other bundles remain unchanged. Hence, the lower bound is preserved after turn $k$.
By induction, $c_i(X_j)\geq b$ for every $j\neq i$ in the final allocation. Since agent $i$ has already been processed and is never chosen as $\ell$ later, her own bundle remains $X_i^i$.
\end{proof}

We can now establish the desired structural property.

\begin{proposition}
\label{prop:generic-pmms}
Algorithm~\ref{alg:main} returns a complete allocation $X$ such that every agent $i$ either receives a singleton or satisfies $c_i(X_i)\leq2c_i(X_j)$ for every $j\neq i$. Consequently, $X$ is $\frac43$-PMMS.
\end{proposition}
\begin{proof}
Completeness follows from Lemma~\ref{lem:phase2-invariant}. Fix an agent $i$. If $i$ is ever chosen as $\ell$ in an exchange, then property~(3) of Lemma~\ref{lem:phase2-invariant} implies that her final bundle $X_i$ is a singleton. By the second bound in Lemma~\ref{lem:basic-bound}, she then satisfies exact PMMS. We may therefore assume that $i$ is never chosen as $\ell$.

\medskip
\noindent\textbf{Case 1: $i\in N_0$.}
Agent $i$ is never processed in Phase~2, so $X_i=Y_i$. By Lemma~\ref{lem:price-bounds}, $c_i(X_i)\leq 2\tau$, while Lemma~\ref{lem:phase2-invariant} gives $c_i(X_j)\geq\tau$ for every $j\neq i$. Hence, $c_i(X_i)\leq 2c_i(X_j)$ for every $j\neq i$.

\medskip
\noindent\textbf{Case 2: $i\in N_H$ and no exchange occurs at her turn.}
When agent $i$ is processed, she holds $S_i\cup\{h_i\}$, whose cost is $s_i+a_i$. Since no exchange occurs, every other current bundle has cost at least $(s_i+a_i)/2$ according to $c_i$. Moreover, $(s_i+a_i)/2<a_i$ because $s_i<a_i$. Applying Lemma~\ref{lem:persistent-lower-bound} with $b=(s_i+a_i)/2$ shows that, in the final allocation, $c_i(X_j)\geq (s_i+a_i)/2$ for every $j\neq i$. Since agent $i$ is never chosen as $\ell$, her own bundle remains unchanged. Therefore, $c_i(X_i)=s_i+a_i\leq 2c_i(X_j)$ for every $j\neq i$.

\medskip
\noindent\textbf{Case 3: $i\in N_H$ and an exchange occurs at her turn.}
Let $B$ be the bundle chosen by agent $i$ for the exchange, and let $y:=c_i(B)$. By Lemma~\ref{lem:phase2-invariant}, $y\geq\tau$, while Lemma~\ref{lem:price-bounds} gives $s_i\leq\tau$. Hence, $s_i\leq y$. Since the exchange condition gives $2y<s_i+a_i$, we obtain $a_i>2y-s_i\geq y$.
After the exchange, agent $i$ receives $S_i\cup B$, whose cost is $s_i+y\leq 2y$. Since $B$ was a least-costly bundle among the other agents, every unchanged bundle held by another agent has cost at least $y$ according to $c_i$, while the new singleton $\{h_i\}$ has cost $a_i>y$. Applying Lemma~\ref{lem:persistent-lower-bound} with $b=y$ shows that, in the final allocation, $c_i(X_j)\geq y$ for every $j\neq i$. Since agent $i$ is never chosen as $\ell$, she keeps $S_i\cup B$. Thus, $c_i(X_i)\leq 2c_i(X_j)$ for every $j\neq i$.

Therefore, every agent either receives a singleton or satisfies $c_i(X_i)\leq 2c_i(X_j)$ for every $j\neq i$. Singleton bundles satisfy exact PMMS by Lemma~\ref{lem:basic-bound}, while Lemma~\ref{lem:2-load-bound} gives the $\frac43$-PMMS guarantee for all remaining agents.
\end{proof}

\subsection{Extension to General Instances}

We now remove the positivity and non-degeneracy assumptions and complete the proof of Theorem~\ref{thm:chores-existence}.

\begin{proof}[Proof of Theorem~\ref{thm:chores-existence}]
Consider an arbitrary additive-chore instance. If $n=1$, there are no pairwise PMMS constraints, so the claim is immediate. Hence, assume $n\geq2$.

We first handle chores that have zero cost to some agent, following \citet[Section~3]{conf/soda/Mahara26}. Let $Z$ be the set of such chores, and let $M^+:=M\setminus Z$. For each $g\in Z$, choose an agent $i$ with $c_i(g)=0$, remove $g$ for now, and assign it to agent $i$ at the end.
It suffices to find a $\frac43$-PMMS allocation of $M^+$. We first observe that $\ChoreMu_i(2,U)\leq\ChoreMu_i(2,V)$ whenever $U\subseteq V$. Indeed, take a partition $(P_1,P_2)\in\PiTwo(V)$ attaining $\ChoreMu_i(2,V)$ and remove the chores in $V\setminus U$ from the two parts. The resulting bundles form a partition of $U$, and neither part becomes more costly. Now consider adding a removed chore $g$ back to its chosen agent $i$. Since $c_i(g)=0$, the cost of agent $i$'s own bundle does not change, while all other agents' bundle costs remain unchanged. Moreover, every pairwise union involving $X_i$ only gains the chore $g$, so its two-way maximin-share benchmark cannot decrease. Thus, adding $g$ back preserves all PMMS inequalities. Repeating this argument for every $g\in Z$, it suffices to allocate the chores in $M^+$.

If $|M^+|\leq n$, assign each chore in $M^+$ to a distinct agent. Every nonempty bundle is then a singleton and satisfies exact PMMS by the second bound in Lemma~\ref{lem:basic-bound}, while empty bundles satisfy exact PMMS trivially. Adding back the chores in $Z$ then proves the claim. We may therefore assume that $|M^+|>n$.

By the definition of $M^+$, every chore in $M^+$ has strictly positive cost to every agent. We next remove the non-degeneracy assumption using the perturbation of \citet[Lemma~2.6]{conf/soda/Mahara26}. Assign a distinct prime $q_{ig}$ to every pair $(i,g)\in N\times M^+$ and, for $\varepsilon>0$, define $c_i^{(\varepsilon)}(g):=c_i(g)q_{ig}^{\varepsilon}$. By the cited lemma, for all sufficiently small $\varepsilon>0$, the perturbed instance is non-degenerate. Moreover, $c_i^{(\varepsilon)}(g)\to c_i(g)$ as $\varepsilon\to0$ for every $i\in N$ and $g\in M^+$.

Choose a sequence $\varepsilon_t>0$ converging to $0$ such that each perturbed instance is non-degenerate. For every $t$, Theorem~\ref{thm:scale}, Algorithm~\ref{alg:main}, and Proposition~\ref{prop:generic-pmms} give a $\frac43$-PMMS allocation $X^{(t)}$ of $M^+$ under the rescaled perturbed costs. By scale invariance, $X^{(t)}$ satisfies the same guarantee under $c^{(\varepsilon_t)}$.
Since there are only finitely many allocations of $M^+$, some allocation $X^*$ occurs along an infinite subsequence of $(X^{(t)})_{t\geq1}$. Passing to this subsequence, we may assume that $X^{(t)}=X^*$ for every
$t$. Hence, for every $i\neq j$,
\[
c_i^{(\varepsilon_t)}(X_i^*)
\leq
\frac43
\min_{(P_1,P_2)\in\PiTwo(X_i^*\cup X_j^*)}
\max\{c_i^{(\varepsilon_t)}(P_1),
      c_i^{(\varepsilon_t)}(P_2)\}.
\]

Since $M^+$ is finite, for every fixed bundle $S\subseteq M^+$, $c_i^{(\varepsilon_t)}(S)\to c_i(S)$ as $t\to\infty$. Moreover, $\PiTwo(X_i^*\cup X_j^*)$ is finite. Hence, for every fixed partition $(P_1,P_2)\in\PiTwo(X_i^*\cup X_j^*)$,
\[
\max\{c_i^{(\varepsilon_t)}(P_1),
      c_i^{(\varepsilon_t)}(P_2)\}
\to
\max\{c_i(P_1),c_i(P_2)\}.
\]
Therefore,
\[
\min_{(P_1,P_2)\in\PiTwo(X_i^*\cup X_j^*)}
\max\{c_i^{(\varepsilon_t)}(P_1),
      c_i^{(\varepsilon_t)}(P_2)\}
\to
\ChoreMu_i(2,X_i^*\cup X_j^*).
\]
Taking limits in the PMMS inequality above gives
\[
c_i(X_i^*)
\leq
\frac43\ChoreMu_i(2,X_i^*\cup X_j^*)
\]
for every $i\neq j$.
Thus, $X^*$ is a $\frac43$-PMMS allocation of $M^+$. Finally, assign each chore in $Z$ to the agent chosen for it above. As shown above, this preserves all PMMS inequalities. Hence, the resulting allocation is $\frac43$-PMMS for the original instance.
\end{proof}

\section{Conclusion and Open Problems}
We give a polynomial-time reduction from additive chores to additive goods that preserves PMMS existence. This transfers the known nonexistence result for chores to goods. Separate constructions in Appendix~\ref{append:impossibility_ratio} give optimal instance factors of $226/227$ for goods and $1.102065$ for chores. For the positive result, we prove that every additive-chore instance admits a $4/3$-PMMS allocation through a reassignment and exchange procedure starting from a price-supported pEF1 allocation.

Determining the optimal universal approximation guarantees remains open. Taking into account the stronger goods impossibility bound of \citet{golz2026pmms}, the bounds discussed in this paper are
\[
0.7808 \leq \alphaPMMS \leq \frac{78}{79},
\qquad
1.102065 \leq \rhoPMMS \leq \frac43.
\]
Closing these gaps is a natural direction for future work. It would also be interesting to replace the exhaustive verification of the quantitative counterexamples with structural proofs and to construct chore instances with stronger lower bounds.

\section*{Declaration of the Use of AI Tools}
The two quantitative counterexample constructions in Appendix~\ref{append:impossibility_ratio} and their finite certificates were derived with the assistance of OpenAI's GPT-5.6 Sol, which also assisted in developing the $4/3$-PMMS approximation proof and in drafting and typesetting the manuscript. All exhaustive searches underlying the results were independently reproduced using exact arithmetic during the preparation of the manuscript.

The authors have independently verified all mathematical details, simplified and refined the arguments and exposition, and retain full responsibility for the content of the manuscript and any errors therein.
% The authors are responsible for the correctness of all results in this paper. The manuscript has not been peer reviewed.

\newpage
\bibliographystyle{plainnat}
\bibliography{reference}

\appendix
\section{Gapped Impossibility Results}
\label{append:impossibility_ratio}
In this section, we extend the non-existence results in Sect.~\ref{sec:mainImpossibility} to their gapped versions.
For additive goods, we show that there exists an instance where the approximation ratio of PMMS is at most $\frac{226}{227}$.
For additive chores, we show that there exists an instance where the approximation ratio is at least $\frac{220413}{200000}=1.102065$.

\subsection{Impossibility for Additive Goods}
We consider throughout this section the following fair division instance with $n=3$ agents and $m = 9$ goods.

\begin{table}[H]
  \centering
  \small
  \setlength{\tabcolsep}{5.2pt}
  \begin{tabular}{c*{9}{r}}
    \toprule
      & $g_1$ & $g_2$ & $g_3$ & $g_4$ & $g_5$ & $g_6$ & $g_7$ & $g_8$ & $g_9$ \\
    \midrule
    $v_1$ & 80 & 63 & 73 & 77 & 81 & 66 & 72 & 60 & 83 \\
    $v_2$ & 86 & 60 & 74 & 76 & 78 & 75 & 73 & 66 & 80 \\
    $v_3$ & 85 & 60 & 63 & 75 & 80 & 72 & 62 & 61 & 81 \\
    \bottomrule
  \end{tabular}
  \caption{Singleton values in the three-agent, nine-good instance.}
  \label{tab:goods-values}
\end{table}

\begin{theorem}\label{thm:goods-ratio}
There is an instance with three agents and nine goods in which every valuation is strictly positive and additive, but no exact PMMS allocation exists. In addition, the optimum PMMS ratio of the instance is exactly $226/227$.
Consequently,
\[
  \alphaPMMS \le \frac{226}{227}<1.
\]
\end{theorem}

To prove the theorem, we first separate allocations according to their
bundle cardinalities. This reduction is particularly useful because all
singleton values in Table~\ref{tab:goods-values} lie in the narrow
interval $[60,86]$. We begin by ruling out every unbalanced allocation.

\begin{lemma}\label{lem:unbalanced}
Every allocation whose bundle cardinalities are not $(3,3,3)$ has PMMS factor strictly below $226/227$.
\end{lemma}

\begin{proof}
Choose an agent $i$ with a smallest bundle and an agent $j$ with a largest bundle, and write
\[
  a=|X_i|,\qquad b=|X_j|.
\]
Since the three cardinalities sum to nine and are not all equal to three, we have $a\le 2$ and $b\ge a+2$. The $a+b$ goods in $X_i\cup X_j$ can be partitioned into bundles of cardinalities $a+1$ and $b-1$. Both parts contain at least $a+1$ goods, and therefore
\[
  \GoodsMu_i(2,X_i\cup X_j)\ge 60(a+1).
\]
At the same time, $v_i(X_i)\le 86a$. Hence
\[
  \frac{v_i(X_i)}{\GoodsMu_i(2,X_i\cup X_j)}
  \le \frac{86a}{60(a+1)}.
\]
For $a\in\{0,1,2\}$, the largest right-hand side occurs at $a=2$, giving
\[
  \frac{v_i(X_i)}{\GoodsMu_i(2,X_i\cup X_j)}
  \le \frac{172}{180}=\frac{43}{45}<\frac{226}{227}.
\]
\end{proof}

We now consider the remaining balanced allocations with $$|X_1| = |X_2| = |X_3| = 3.$$

There are exactly $$\binom93\binom63=84\cdot20=1680$$
different possible allocations.

The following observation reduces the computation of each pairwise
benchmark to only ten candidate partitions.

\begin{lemma}\label{lem:balanced}
    Fix a balanced allocation and an ordered pair $i \neq j$. The union $X_i \cup X_j$ contains six goods, and $\mu_i^G(2,X_i\cup X_j)$ is attained by a 3-versus-3 bipartition.
\end{lemma}

\begin{proof}
    By definition, every three-good bundle has value at least $3 \cdot 60 = 180$ to agent $i$. Hence $$\mu_i^G(2, X_i \cup X_j) \geq 180.$$
    Any bipartition with a side of cardinality at most two has smaller-side value at most $2 \cdot 86 = 172$. Such a partition cannot attain a benchmark of at least $180$. Therefore, every maximizing bipartition must have three goods on each side. Consequently, only the $$\frac12 \binom63=10$$ unordered 3-versus-3 partitions can certify a violation at the target ratio.
\end{proof}

\begin{proposition}\label{prop:balanced}
Every balanced allocation has PMMS ratio at most $226/227$ and there exists an allocation that attains it exactly.
\end{proposition}
\begin{proof}
We first show that no balanced allocation can have PMMS factor strictly
larger than $226/227$. We exhaustively check all $1680$ balanced
allocations. For each such allocation $X$, the verification finds an
ordered pair $i\ne j$ and a $3$-versus-$3$ partition
\[
  X_i\cup X_j=P_1\sqcup P_2
\]
such that
\[
  227\,v_i(X_i)
  \le
  226\,\min\{v_i(P_1),v_i(P_2)\}.
\]
Since
\[
  \GoodsMu_i(2,X_i\cup X_j)
  \ge
  \min\{v_i(P_1),v_i(P_2)\},
\]
it follows that
\[
  \frac{v_i(X_i)}
       {\GoodsMu_i(2,X_i\cup X_j)}
  \le
  \frac{226}{227}.
\]
Hence $\alpha(X)\le 226/227$ for every balanced allocation $X$.

It remains to show that this bound is attainable. Consider the allocation
\[
  X_1=\{g_2,g_3,g_9\},\qquad
  X_2=\{g_5,g_6,g_7\},\qquad
  X_3=\{g_1,g_4,g_8\}.
\]
The six directed pairwise benchmarks are shown in
Table~\ref{tab:goods-witness}.

\begin{table}[h]
  \centering
  \small
  \begin{tabular}{cccc}
    \toprule
    Ordered pair
      & $v_i(X_i)$
      & $\GoodsMu_i(2,X_i\cup X_j)$
      & Ratio \\
    \midrule
    $1\rightarrow 2$ & $219$ & $219$ & $1$ \\
    $1\rightarrow 3$ & $219$ & $217$ & $219/217$ \\
    $2\rightarrow 1$ & $226$ & $218$ & $226/218$ \\
    $2\rightarrow 3$ & $226$ & $227$ & $226/227$ \\
    $3\rightarrow 1$ & $221$ & $209$ & $221/209$ \\
    $3\rightarrow 2$ & $221$ & $217$ & $221/217$ \\
    \bottomrule
  \end{tabular}
  \caption{Directed pairwise benchmarks for a balanced allocation
  attaining the goods bound.}
  \label{tab:goods-witness}
\end{table}

For the critical ordered pair $2\rightarrow 3$,
\[
  X_2\cup X_3
  =
  \{g_1,g_4,g_5,g_6,g_7,g_8\}.
\]
Agent $2$ values this union at
\[
  86+76+78+75+73+66=454.
\]
The partition
\[
  \{g_1,g_6,g_8\}
  \sqcup
  \{g_4,g_5,g_7\}
\]
has values
\[
  86+75+66=227
  \qquad\text{and}\qquad
  76+78+73=227,
\]
respectively. Therefore
\[
  \GoodsMu_2(2,X_2\cup X_3)\ge 227.
\]
On the other hand, in any bipartition of a set of total value $454$,
the less valuable side has value at most $454/2=227$. Hence
\[
  \GoodsMu_2(2,X_2\cup X_3)=227.
\]
Since
\[
  v_2(X_2)=78+75+73=226,
\]
the corresponding directed ratio is exactly
\[
  \frac{v_2(X_2)}
       {\GoodsMu_2(2,X_2\cup X_3)}
  =
  \frac{226}{227}.
\]
Every other directed ratio in Table~\ref{tab:goods-witness} is at least
$1$. Consequently,
\[
  \alpha(X)=\frac{226}{227}.
\]
\end{proof}

Lemma~\ref{lem:unbalanced} and Proposition~\ref{prop:balanced} give the universal upper bound for this instance. This proves Theorem~\ref{thm:goods}.

\paragraph{Exact verification.}
The proof above uses Lemma~\ref{lem:unbalanced} to eliminate unbalanced
allocations analytically and therefore requires an exhaustive check only
of the $1680$ balanced allocations. As an independent audit, the exact
verifier
\path{pmms_226_227_bundle/check_226_227.py}
instead enumerates all
$3^9=19{,}683$ labeled allocations of the nine goods among the three agents, without
using the cardinality reduction.

For every allocation, the verifier computes all six directed two-way
maximin-share benchmarks exactly and hence evaluates $\alpha(X)$. The
full enumeration returns
\[
\begin{array}{lr}
\text{balanced allocations}                         & 1{,}680,\\
\text{unbalanced allocations}                       & 18{,}003,\\
\text{exact-PMMS allocations}                       & 0,\\
\text{maximum unbalanced factor}                    & 83/103,\\
\text{maximum balanced factor}                      & 226/227,\\
\text{allocations attaining the optimum}            & 3,\\
\text{balanced allocations without a certificate}  & 0.
\end{array}
\]
In particular, the unrestricted exhaustive computation independently
confirms that
\[
  \max_X\alpha(X)
  =
  \frac{226}{227}.
\]

All singleton values and two-way maximin-share benchmarks are integers,
and all comparisons of PMMS factors are performed by exact rational
cross-multiplication. Thus the verification uses no floating-point
comparisons.

\subsection{Impossibility for Additive Chores}
We consider throughout this section the following fair division instance with $n=5$ agents and $m=12$ chores. The chores have three types, denoted by $A,B,C$, with four identical copies of each type.
\begin{table}[H]
  \centering
  \small
  \begin{tabular}{c*{3}{r}}
    \toprule
      & $A$ & $B$ & $C$ \\
    \midrule
    $a_1,a_2$       & 140   & 100    & 21 \\
    $a_3,a_4,a_5$   & 20413 & 157283 & 100000 \\
    \bottomrule
  \end{tabular}
  \caption{Singleton costs in the five-agent, twelve-chore instance.}
  \label{tab:chores-values}
\end{table}
\begin{theorem}\label{thm:chores}
There is an instance with five agents and twelve chores in which every cost
function is strictly positive and additive, but no exact PMMS allocation
exists. In addition, the optimum PMMS ratio of the instance is exactly
\[
  R := \frac{220413}{200000}=1.102065.
\]
Consequently,
\[
  \rhoPMMS\ge R>1.
\]
\end{theorem}
The proof has two parts. First we exhibit an allocation whose factor is exactly $R$. Second, an exact exhaustive verifier proves that every allocation has factor at least $R$.

\paragraph{Type-count reduction.} Since all chores of the same type have identical costs to every agent, permuting copies within a chore type does not change any bundle cost or pairwise maximin-share benchmark. Hence an allocation can be represented by the type-count vector of each agent's bundle.

For each chore type, distributing four identical copies among five labeled
agents is equivalent to choosing a weak composition of $4$ into $5$ parts,
of which there are
\[
  \binom84=70.
\]
Since the three chore types can be distributed independently, exact
enumeration reduces to
\[
  70^3=343{,}000
\]
type-count allocations. The agents remain labeled throughout; only
permutations among chores of the same type are removed.

For a type-count vector $(a,b,c)$, we write
\[
  c_i(a,b,c)
  :=
  a\,c_i(A)+b\,c_i(B)+c\,c_i(C).
\]

\begin{proposition}\label{prop:chore-attain}
    The chore instance has an allocation with PMMS factor exactly $R$.
\end{proposition}

\begin{proof}
    Consider the type-count allocation
    \[
    X_1 = (0,1,1), \qquad X_2 = (0,1,1), \qquad X_3 = (1,0,2), \qquad X_4 = (1,1,0), \qquad X_5 = (2,1,0).
    \]
    The coordinates sum to $(4,4,4)$, so the allocation is feasible.

    For the ordered pair $3\to4$, we have

$$
  c_3(X_3)=20413+2(100000)=220413,
$$

and

$$
  X_3\cup X_4=(2,1,2).
$$

The bipartition

$$
  2C
  \qquad\text{versus}\qquad
  2A+B
$$

has costs

$$
  200000
  \qquad\text{and}\qquad
  2(20413)+157283=198109,
$$

so

$$
  \ChoreMu_3(2,X_3\cup X_4)\le 200000.
$$

The total cost of $(2,1,2)$ is $398109$. Hence, if a bipartition had maximum side-cost strictly below $200000$, one side would have cost strictly between $198109$ and $200000$. A direct check of the subbundles of $(2,1,2)$ shows that no such cost occurs. Therefore

$$
  \ChoreMu_3(2,X_3\cup X_4)=200000,
$$

and thus

$$
  \frac{c_3(X_3)}
       {\ChoreMu_3(2,X_3\cup X_4)}
  =
  \frac{220413}{200000}
  =
  R.
$$

For the remaining nineteen ordered pairs, direct evaluation gives

$$
  \max_{(i,j)\ne(3,4)}
  \frac{c_i(X_i)}
       {\ChoreMu_i(2,X_i\cup X_j)}
  =
  \frac{220413}{218522}
  <
  R.
$$

Hence

$$
  \rho(X)=R.
$$
\end{proof}

\begin{proposition}\label{prop:chore-lower}
    Every allocation of the chore instance has PMMS factor at least $R$.
\end{proposition}
\begin{proof}
    By the type-count reduction above, it suffices to consider the \[70^3 = 343{,}000\] labeled type-count allocations.

    Fix such an allocation $X$. For every ordered pair $i \neq j$, let \[
    U = X_i \cup X_j
    \] denote the type-count vector of the pairwise union. The checker computes \[
    \mu_i^C(2,U) = \min_{P \subseteq U} \max \{c_i(P),c_i(U \setminus P)\}.
    \]
    For every allocation $X$, the checker verifies that there exists an ordered pair $i \neq j$ with positive benchmark such that \[
    \frac{c_i(X_i)}{\mu_i^C(2,X_i \cup X_j)} \geq \frac{220413}{200000} = R.
    \] 
    Equivalently, it verifies the integer inequality \[200000 c_i(X_i) \geq 220413 \mu_i^C(2, X_i \cup X_j).\]
    Hence, every allocation satisfies \[
    \rho(X) = \max_{i \neq j} \frac{c_i(X_i)}{\mu_i^C(2,X_i \cup X_j)} \geq R.
    \]
\end{proof}

\paragraph{Exact enumeration.} The exhaustive search over all $343{,}000$ labeled type-count allocations returns
$$
\begin{array}{lr}
\text{type-count allocations checked}        & 343{,}000,\\
\text{minimum PMMS factor}                    & 220413/200000,\\
\text{allocations attaining the minimum}      & 36.
\end{array}
$$
Together with Proposition~\ref{prop:chore-attain} and Proposition~\ref{prop:chore-lower}, this establishes \[
\min_X \rho(X) = \frac{220413}{200000} = 1.102065.
\]
The verification is entirely exact. All singleton costs, bundle costs, and two-way maximin-share benchmarks are integers. Comparisons of PMMS ratios are performed by integer cross multiplication; in particular, testing whether a directed ratio is at least $R$ amounts to checking \[
200000 c_i(X_i) \geq 220413 \mu_i^C(2, X_i \cup X_j).
\]
Consequently, the exhaustive certificate uses no floating-point arithmetic or numerical tolerance.

\end{document}